\documentclass[preprint,11pt,3p]{elsarticleX}
\usepackage[utf8]{inputenc}

\usepackage[T1]{fontenc}
\usepackage{utopia}
\usepackage{amssymb}
\usepackage{amsthm}
\usepackage{amsmath}

\usepackage{color}
\usepackage{subfig}
\usepackage{wrapfig}

\usepackage[mathlines]{lineno}

\newcommand{\M}{\ensuremath{\mathcal{M}}}

\newtheorem {theorem} {Theorem}

\newtheorem {lemma} {Lemma}

\journal{Discrete Mathematics}

\begin{document}

\begin{frontmatter}

%\title{Matching points with diametral disks \tnoteref{label0}}
%\tnotetext[label0]{This is only an example}
%\title{Matching points with diametral disks}
%\title{Matching points with common intersection disks}
\title{Matching points with disks with a common intersection}

\author[label1]{Clemens Huemer}
\address[label1]{Departament de Matem\`atiques, Universitat Polit\`ecnica de Catalunya, Spain}

% ..... \corref{cor1}
%\cortext[cor1]{I am corresponding author}

\ead{clemens.huemer@upc.edu}

\author[label2]{Pablo P\'erez-Lantero}
\address[label2]{Departamento de Matem\'atica y Ciencia de la Computaci\'on, Universidad de Santiago, Chile}
\ead{pablo.perez.l@usach.cl}

\author[label1]{Carlos Seara}
\ead{carlos.seara@upc.edu}

\author[label1]{Rodrigo I. Silveira}
\ead{rodrigo.silveira@upc.edu}

\begin{abstract}
%We consider matchings between a set $R$ of red points and a set $B$ of blue points with diametral disks. 
We consider matchings with diametral disks between two sets of points $R$ and $B$.
More precisely, 
for each pair of matched points $p \in R$ and $q \in B$, we consider the disk through $p$ and $q$ with the smallest diameter.
We prove that for any $R$ and $B$ such that $|R|=|B|$, there exists a perfect matching such that the diametral disks of the matched point pairs
have a common intersection. In fact, our result is stronger, and shows that a maximum weight perfect matching has this property.
\end{abstract}

\end{frontmatter}

% \linenumbers


\section{Introduction}

We consider two sets of $n\ge 2$ points in the plane, $R$ and $B$, that are assumed to be disjoint.
We call the points in $R$ \emph{red}, and those in $B$  \emph{blue}.
A well-known family of problems involving red and blue points is that of matching points with (pairwise disjoint) geometric objects.
The goal is to find pairs of points such that each pair is associated with a geometric object that covers both points of the pair, 
and all associated objects are pairwise disjoint. In each pair the two points are restricted to be of different colors, 
or in each pair the two points are restricted to be of the same color.
This class of problems is well studied in discrete and computational geometry, starting from the classic result that $n$ red points and $n$ blue points
can always be perfectly matched with $n$ pairwise non-crossing segments, where each segment connects a red point with a blue point~\cite{larson1983problem}.
The study has been continued in plenty of directions, for both the monochromatic and bichromatic versions, 
by using pairwise disjoint segments~\cite{aloupis2015,dumitrescu2001}, rectangles and squares~\cite{abrego2004matching,abrego2009,bereg2009,caraballo2014matching},
and more general geometric objects~\cite{AloupisCCDDDMHHLSST13}.

More formally, for $R=\{p_1,\dots,p_n\}$ and $B=\{q_1,\dots,q_n\}$, a {\em matching} of $R\cup B$ is a partition of $R\cup B$ into $n$ pairs such that each pair consists of a red and a blue point. A point $p\in R$ and a point $q\in B$ are {\em matched} if and only if the pair $(p,q)$ is in the
matching.

We use $pq$ to denote the segment connecting $p$ and $q$, and $|pq|$ to denote its length.
The {\em diametral disk} of $pq$, denoted $D_{pq}$, is the disk with diameter equal to $|pq|$ that is centered at the midpoint of $pq$, while  $\mathcal{C}_{pq}$ is the corresponding circle.
  For a matching $\M$, we use $D_\M$ to denote the set of disks associated with the matching, that is: $D_\M = \{D_{pq} \mid (p,q) \in \M \}$.

\begin{wrapfigure}{r}{0.4\textwidth}
%\begin{figure}[t]
	\centering
	\includegraphics[scale=1,page=4]{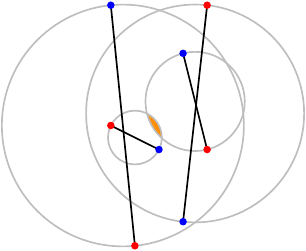}
	\caption{\small{
		Example for a set of $n=4$ red and blue points,
		showing a matching and the associated disks,
		which have a common intersection.}}
	\label{fig:example}
%\end{figure}
\end{wrapfigure}
In this paper, we prove that for any $R$ and $B$ as above, there always exists a matching $\M$ such that all disks in $D_\M$ have a common intersection (see Figure~\ref{fig:example}). More precisely, we show that any \emph{maximum matching} satisfies this property. A matching $\M$ of $R\cup B$ is {\em maximum} if it maximizes the sum of the squared distances between the matched points, that is, it maximizes $\sum_{(p,q)\in\M}|pq|^2$.

Observe that our result goes in the direction opposite to that of known results on matching red and blue points: Our goal is that all matching objects have a common intersection, whereas in
previous work (e.g., \cite{aloupis2015,dumitrescu2001,abrego2004matching,abrego2009,bereg2009,caraballo2014matching,AloupisCCDDDMHHLSST13}, and the references in~\cite{AloupisCCDDDMHHLSST13}) it is  required that all matching objects are pairwise disjoint.

Moreover, in order for the problem to make sense, the object used for the matching is important.
Segments---arguably the simplest geometric object defined by two points---do not work.
That is, when matching points of different colors with segments, it is not always possible to guarantee that all matching segments are pairwise intersecting (e.g.,\ consider two red and two blue points not in convex position).
This also happens when matching with axis-aligned rectangles, where for any two matched points the associated rectangle is the one with minimum-area that contains both points.
For these reasons, we focus on matching with disks.
The choice of diametral disks comes from the need of bounding the size of the disks. Otherwise, if disks can be arbitrarily large, the result becomes trivially true.
Furthermore, diametral disks are a natural choice for disks that must be defined by two points.

\begin{wrapfigure}{r}{0.4\textwidth}
%\begin{figure}[t]
	\centering
	\includegraphics[scale=0.3,trim=5cm 9cm 11cm 4cm, clip]{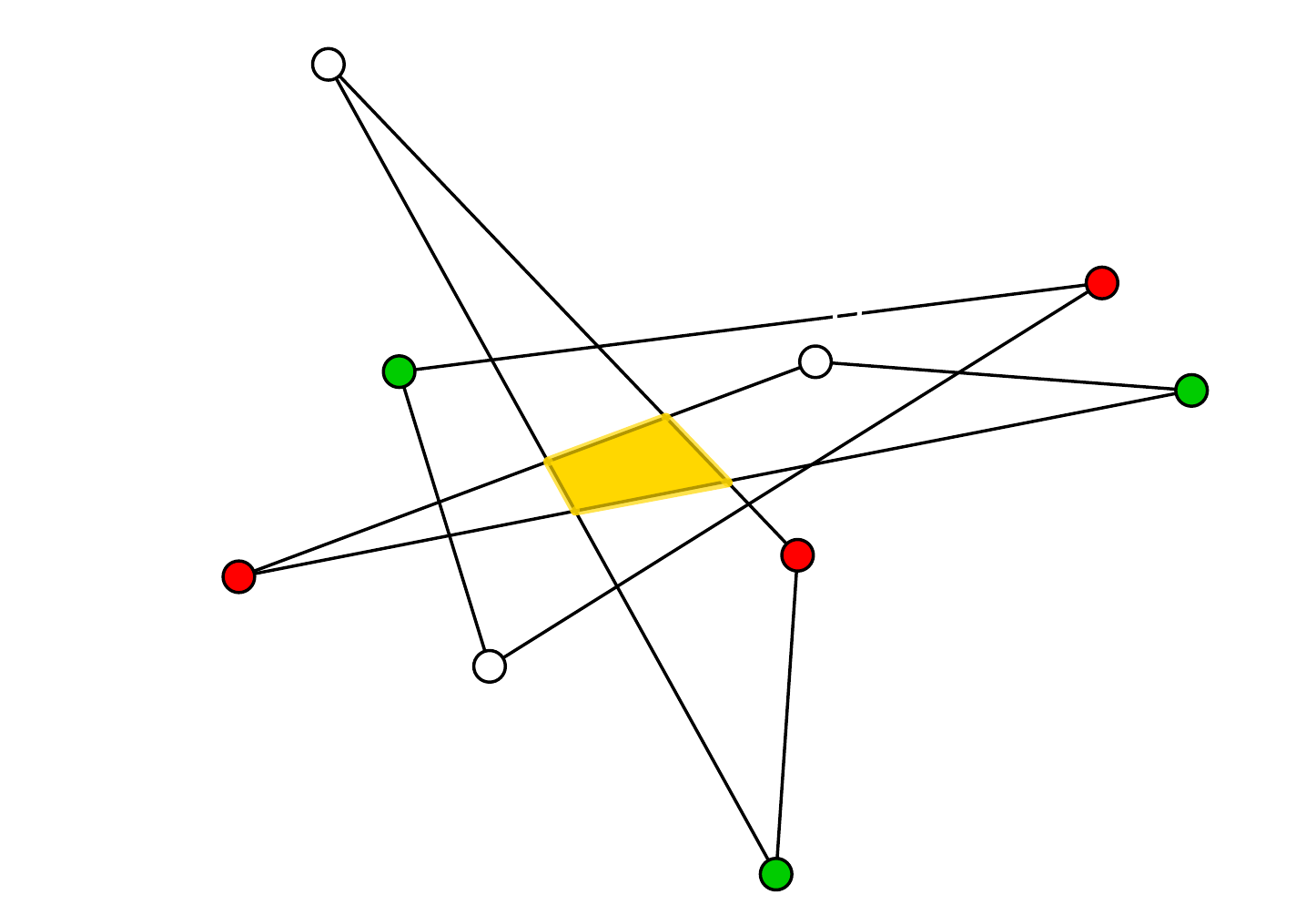}
	\caption{\small{
		Example for a set of $n=3$ red, white, and green points,
		showing a partition into three-colored triangles,
		which have a common intersection.}}
	\label{fig:trianglesintersect}
%\end{figure}
\end{wrapfigure}
Our result is also motivated from Tverberg's theorem~\cite{tverberg1966}, which states that any $(d+1)(r-1)+1$ points in $\mathbb{R}^d$ can be partitioned 
into $r$ subsets such that the convex hulls of the subsets have a point in common. A colored version of this theorem for points in the plane was 
proved in~\cite{barany2}, implying that any set of $n$ red, $n$ white, and $n$ green points in the plane can be partitioned into $n$ triangles, with one vertex from each color, 
such that the triangles have a point in common. See Figure~\ref{fig:trianglesintersect}. Many extensions and related results are known, 
see for instance~\cite{arocha,barany,holmsen,loera}.

Here we consider diametral disks instead of convex hulls. Hence, our result can be put into the context of Tverberg type theorems.
Other related results, albeit for different problems, include those in~\cite{SmorodinskyS04,AshokGR16,akiyama,barany3,prodromou}.

%\clemens{Creo conviene buscar un poco más de literatura relacionada, por ejemplo del estilo:\\
%-Selection Lemmas for Various Geometric Objects by Pradeesha Ashok, Sathish Govindarajan, Ninad Rajgopal\\
%-Selecting Heavily Covered Points by Pseudo-circles, Spheres and Rectangles by Shakhar Smorodinsky , Micha Sharir }

%\reviewer{I think the problem should be a bit better motivated since the goal, as the authors say, goes in the direction opposite to all known results on the topic.}

\medskip
\noindent {\bf Outline.} We begin by introducing some additional notation. After that, we consider a maximum matching $\M$ of $R\cup B$, and prove in Section~\ref{sec:proofs_1} that any pair of disks in $D_\M$ intersect. Finally, in Section~\ref{sec:proofs_2}, we prove that all disks in $D_\M$ must intersect.

\medskip
\noindent {\bf Notation.} For a point $p$, let $x(p)$ and $y(p)$ denote the $x$- and $y$-coordinates of $p$, respectively. Given three different points $p$, $q$, and $r$, let $\ell(p,q)$ denote the line containing both $p$ and $q$, $\Delta pqr$ the triangle with vertex set $\{p,q,r\}$, $\angle pqr$ the angle at $q$ in the triangle $\Delta pqr$.
Finally, we will say that a set of points is in \emph{general position} if no three points of the set are collinear.

\section{Any two disks in $D_\M$ intersect }
\label{sec:proofs_1}

We begin by showing in this section that in a maximum matching \M, any pair of disks in  $D_\M$  intersect. To that end, we first prove the following auxiliary result that concerns only four points.
This result will also be the key for proving our main technical result, Lemma~\ref{lem:3matching}.

\begin{lemma}\label{lem:2-matching}
Let $p_1,p_2\in R$ and $q_1,q_2\in B$ such that $\{(p_1,q_1),(p_2,q_2)\}$ is a maximum matching for $\{p_1,p_2,q_1,q_2\}$. Suppose further that $y(p_1)=y(p_2)$, and $x(p_1) < x(p_2)$. Then, $x(q_2) \leq x(q_1)$.
\end{lemma}

\begin{proof}
Assume w.l.o.g. that $p_1=(-1,0)$ and $p_2=(1,0)$.
Refer to Figure~\ref{fig:2-matching}(a).
 Given a constant $c$, the points $r=(x,y)$ that satisfy $|rp_1|^2-|rp_2|^2=c$ are those such that $(x+1)^2+y^2-(x-1)^2 - y^2=c$, which is equivalent to $4x=c$. Then, the locus of such points is the vertical line $x=c/4$. Since $\{(p_1,q_1),(p_2,q_2)\}$ is a maximum matching, we have that
\[
    |p_1q_1|^2 + |p_2q_2|^2  \ge |p_1q_2|^2 + |p_2q_1|^2  \iff
    |p_1q_1|^2 - |p_2q_1|^2  \ge |p_1q_2|^2 - |p_2q_2|^2.
  \]

Let $d_1=|p_1q_1|^2 - |p_2q_1|^2$ and  $d_2=|p_1q_2|^2 - |p_2q_2|^2$. Note that the vertical line through $q_1$ is the line $x=d_1/4$, thus $x(q_1)=d_1/4$, and analogously, $x(q_2)=d_2/4$. Since $d_2 \le d_1$, we have $x(q_2) \leq x(q_1)$.
%\reviewer{I think it suffices to give the second inequality (to have some extra space if needed}
\end{proof}

\begin{figure}[t]
	\centering
	\includegraphics{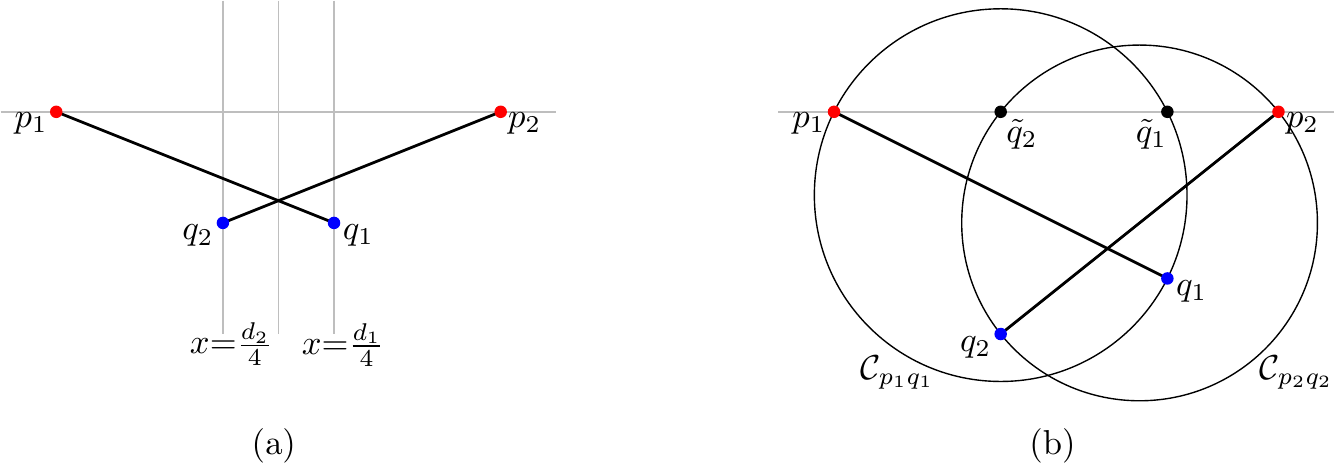}
	\caption{\small{
		(a) Illustration of Lemma~\ref{lem:2-matching}. (b) Illustration of Lemma~\ref{lem:pairwise-intersection}.
	}}
	\label{fig:2-matching}
\end{figure}

Now we can prove that in a maximum matching for four points, the two disks intersect.

\begin{lemma}\label{lem:pairwise-intersection}
Let $p_1,p_2\in R$ and $q_1,q_2\in B$ such that $\{(p_1,q_1),(p_2,q_2)\}$ is a maximum matching for $\{p_1,p_2,q_1,q_2\}$. Then, $D_{p_1q_1}\cap D_{p_2q_2}\neq \emptyset$.
\end{lemma}

\begin{proof}
Assume w.l.o.g.\ that $y(p_1)=y(p_2)$ and $x(p_1) < x(p_2)$.
Refer to Figure~\ref{fig:2-matching}(b).
Let $\tilde{q_1}$ and $\tilde{q_2}$ be the orthogonal projections of $q_1$ and $q_2$ on $\ell(p_1,p_2)$, respectively. By Thales' theorem, $\tilde{q_1}$ lies on $\mathcal{C}_{p_1q_1}$, which implies $p_1\tilde{q_1}=D_{p_1q_1}\cap\ell(p_1,p_2)$. Similarly,  $\tilde{q_2}$ lies on $\mathcal{C}_{p_2q_2}$, and $p_2\tilde{q_2}= D_{p_2q_2}\cap\ell(p_1,p_2)$. By Lemma~\ref{lem:2-matching}, $x(q_2) \leq x(q_1)$, which implies that segments $p_1\tilde{q_1}$ and $p_2\tilde{q_2}$ have a point in common. Hence, $D_{p_1q_1}\cap D_{p_2q_2}\neq \emptyset$.
\end{proof}

It remains to extend the previous result to $n$ points.
To that end, observe that in any maximum matching $\M$ of $R\cup B$, where $|R|=|B|\ge 2$, $\{(p_1,q_1),(p_2,q_2)\}$ is a maximum matching of $\{p_1,p_2,q_1,q_2\}$ for every pair $(p_1,q_1),(p_2,q_2)\in\M$.
That is, any two pairs of a maximum matching form also a maximum matching for the four points involved. Therefore, applying Lemma~\ref{lem:pairwise-intersection} we can conclude that in any maximum matching  $\M$ the disks $D_\M$ are pairwise intersecting.

\section{All disks in $D_\M$ intersect}
\label{sec:proofs_2}

The main goal of this section is to generalize the result in Lemma~\ref{lem:pairwise-intersection} from four to six points.
That is, we will consider sets of three disks from a maximum matching, and will show in Lemma~\ref{lem:3matching} that we can always shrink the disks until finding a point in common.
Then Helly's theorem will imply our main result.
However, this will require considerably more effort and the help of several geometric observations.
\begin{figure}[t]
	\centering
	\includegraphics[scale=1,page=1]{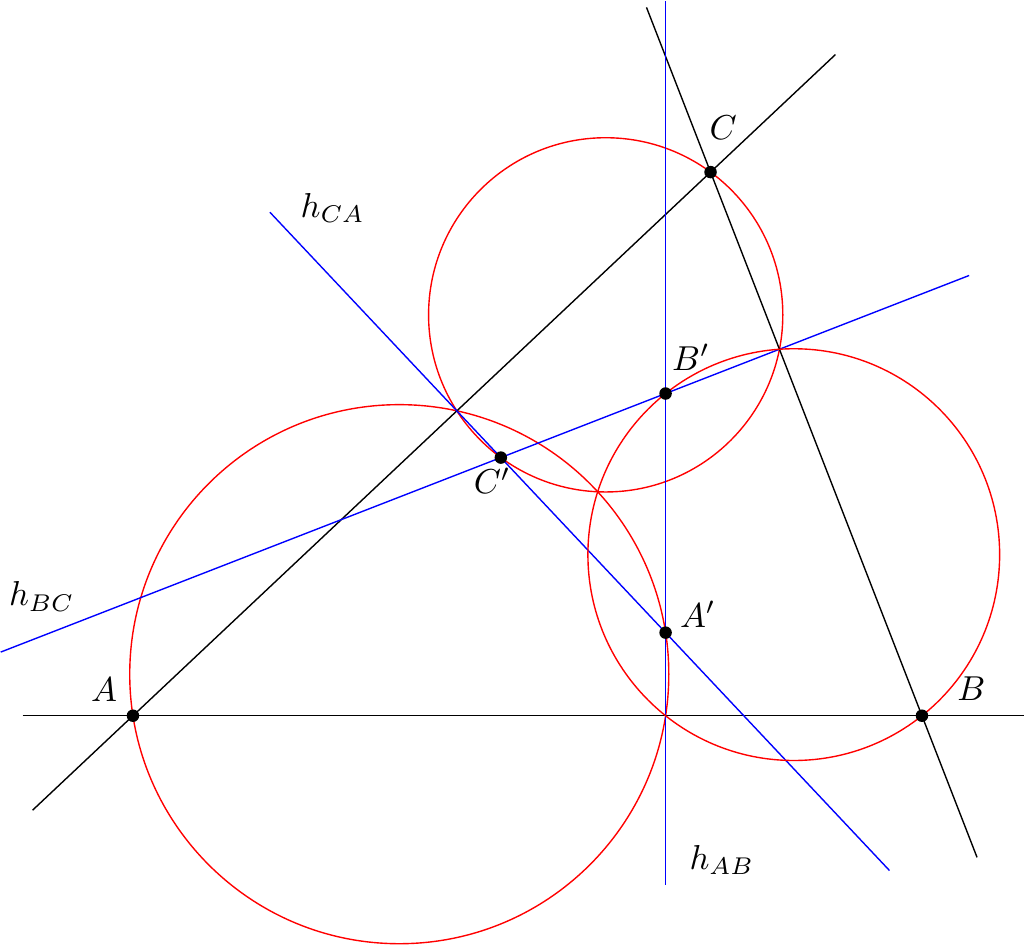}
	\caption{\small{
		Illustration of Lemma~\ref{lem:3circles}. The lemma states that for any three points in general position $A,B,C$, one can define three points $A',B',C'$ based on certain perpendicular lines, such that the associated diametral disks (shown red) intersect at one point.
	}}
	\label{fig:3circles1}
\end{figure}
The next three lemmas describe three different geometric situations at which we will arrive in the proof of   Lemma~\ref{lem:3matching}.
The first lemma is illustrated in  Figure~\ref{fig:3circles1}.

\begin{lemma}\label{lem:3circles}
Let $A$, $B$, and $C$ be three points in the plane in general position. Let $h_{AB}$, $h_{BC}$, and $h_{CA}$ be three lines that are perpendicular to $\ell(A,B)$, $\ell(B,C)$, and $\ell(C,A)$, respectively. Let the points $A'=h_{AB}\cap h_{CA}$, $B'=h_{AB}\cap h_{BC}$, and $C'=h_{BC}\cap h_{CA}$. Then, the three circles $\mathcal{C}_{AA'}$, $\mathcal{C}_{BB'}$, and $\mathcal{C}_{CC'}$ intersect at one point.
\end{lemma}

\begin{proof}
Without loss of generality assume that $A=(a,0)$, $B=(b,0)$, and $C=(0,c)$, for some $a<0$, and $b,c>0$. Observe that triangles $\Delta ABC$ and $\Delta A'B'C'$ are similar, so that $\Delta A'B'C'$ is obtained from $\Delta ABC$ by a rotation of $\pi/2$ radians, a scaling of factor $\lambda$, for some $\lambda > 0$, and finally a translation by some vector $(\alpha,\beta) \in\mathbb{R}^2$.
Assume that the rotation is counter-clockwise (the clockwise case is analogous).

Then we have $A'=\lambda \cdot(0,a)+(\alpha,\beta)=(\alpha,\lambda a+\beta)$, $B'=\lambda \cdot(0,b)+(\alpha,\beta)=(\alpha,\lambda b+\beta)$, and $C'=\lambda \cdot(-c,0)+(\alpha,\beta)=(-\lambda c+\alpha,\beta)$. The points $(x,y)$ of $\mathcal{C}_{AA'}$ are those such that the scalar product between vectors $(x,y)-A=(x-a,y)$ and $(x,y)-A'=(x-\alpha,y-\lambda a -\beta)$ equals zero. That is,
\begin{equation}\label{eq1}
    (x-a)(x-\alpha) + y(y-\lambda a -\beta) ~=~ 0.
\end{equation}
Similarly, the points $(x,y)$ of $\mathcal{C}_{BB'}$ satisfy that the scalar product between $(x,y)-B=(x-b,y)$ and $(x,y)-B'=(x-\alpha,y-\lambda b -\beta)$ equals zero. That is,
\begin{equation}\label{eq2}
    (x-b)(x-\alpha) + y(y-\lambda b -\beta) ~=~ 0.
\end{equation}
One solution to the system formed by equations~\eqref{eq1} and~\eqref{eq2} is the point $(\alpha,0)=h_{AB}\cap \ell(A,B)$, which is one of the intersection points between $C_{AA'}$ and $C_{BB'}$. The other intersection point (considering multiplicity) can be found as follows. Subtracting~\eqref{eq2} from~\eqref{eq1}:
\begin{eqnarray}
    \nonumber
    (b-a)(x-\alpha) + y(\lambda(b-a)) & = & 0 \\
            \label{eq3}
            x & = & -\lambda y + \alpha.
\end{eqnarray}
Substituting equation~\eqref{eq3} in equation~\eqref{eq1}, we obtain
\begin{eqnarray*}
    (-\lambda y + \alpha - a)(-\lambda y) + y(y-\lambda a -\beta) & = & 0 \\
%    y(\lambda^2 y - \lambda\alpha + \lambda a + y -\lambda a - \beta) & = & 0 \\
    y(\lambda^2 y + y - \lambda\alpha - \beta) & = & 0 \\
    y & = & \frac{\lambda\alpha + \beta}{1+\lambda^2}.
\end{eqnarray*}
Then,
\begin{align}
    x &=~ -\lambda\left( \frac{\lambda\alpha + \beta}{1+\lambda^2}\right) + \alpha
      =~ \frac{-\lambda^2\alpha - \lambda\beta + \alpha + \lambda^2\alpha}{1+\lambda^2}
      ~=~ \frac{- \lambda\beta + \alpha}{1+\lambda^2}.
\end{align}

The points $(x,y)$ of $\mathcal{C}_{CC'}$ satisfy that the scalar product between vectors $(x,y)-C=(x,y-c)$ and $(x,y)-C'=(x+\lambda c -\alpha,y -\beta)$ equals zero. That is,
\begin{equation}\label{eq4}
    x(x+\lambda c -\alpha) + (y-c)(y-\beta) ~=~ 0.
\end{equation}
To show the lemma it suffices to prove that $(x,y)=(\frac{- \lambda\beta + \alpha}{1+\lambda^2},\frac{\lambda\alpha + \beta}{1+\lambda^2})$ satisfies equation~\eqref{eq4}.

\begin{align*}
 x(x+\lambda c -\alpha)     & = \left(\frac{- \lambda\beta + \alpha}{1+\lambda^2}\right)
                                  \left(\frac{- \lambda\beta + \alpha}{1+\lambda^2}+\lambda c -\alpha\right) \\
%         & = \left(\frac{- \lambda\beta + \alpha}{1+\lambda^2}\right)
%               \left(\frac{- \lambda\beta + \alpha + \lambda c + \lambda^3c -\alpha - \lambda^2\alpha}{1+\lambda^2}\right) \\
        & = -\lambda\left(\frac{- \lambda\beta + \alpha}{1+\lambda^2}\right)
                \left(\frac{\lambda\alpha + \beta - c - \lambda^2c}{1+\lambda^2}\right) \\
    (y-c)(y-\beta) & =  \left(\frac{\lambda\alpha + \beta}{1+\lambda^2}-c\right)
                         \left(\frac{\lambda\alpha + \beta}{1+\lambda^2}-\beta\right) \\
%          & =  \left(\frac{\lambda\alpha + \beta - c - \lambda^2c}{1+\lambda^2}\right)
%                \left(\frac{\lambda\alpha + \beta - \beta - \lambda^2\beta}{1+\lambda^2}\right) \\
          & = \lambda\left(\frac{\lambda\alpha + \beta - c - \lambda^2c}{1+\lambda^2}\right)
               \left(\frac{- \lambda\beta+\alpha}{1+\lambda^2}\right) \\
         & = -x(x+\lambda c -\alpha).
\end{align*}
Therefore, $(x,y)=(\frac{- \lambda\beta + \alpha}{1+\lambda^2},\frac{\lambda\alpha + \beta}{1+\lambda^2})$ satisfies equation~\eqref{eq4} and is common to $\mathcal{C}_{AA'}$, $\mathcal{C}_{BB'}$, and $\mathcal{C}_{CC'}$.
\end{proof}

We continue with the following lemma that describes a situation on four points. 
Refer to Figure~\ref{fig:aux1}.

\begin{lemma}\label{lem:aux1}
Let $A$, $B$, $P$, and $R$ be four points in the plane such that $\ell(A,B)$ is horizontal, $B$ is to the right of $A$, $P$ belongs to $\ell(A,B)$, and $R$ is above $\ell(A,B)$. Let $\mathcal{C}_1$ be the circle through the points $A$, $P$, and $R$, and $\mathcal{C}_2$ be a circle through $B$ and $P$. If $\mathcal{C}_1$ and $\mathcal{C}_2$ are tangent, let $O=P$, otherwise let $O$ be the  intersection point different from $P$ between $\mathcal{C}_1$ and $\mathcal{C}_2$. Then, if $\mathcal{C}_2$ does not enclose $R$, the points $O$ and $B$ are in the same side of $\ell(A, R)$.
\end{lemma}

\begin{proof}
Consider the case where $P$ is to the right of $A$ (see Figure~\ref{fig:aux1-1}). Make a circle inversion at $A$ (with any radius), and let $B'$, $P'$, $R'$, $O'$, $\mathcal{C}'_1$, and $\mathcal{C}'_2$ denote the images of $B$, $P$, $R$, $O$, $\mathcal{C}_1$, and $\mathcal{C}_2$, respectively (see Figure~\ref{fig:aux1-2}). Note that $\mathcal{C}'_2$ is the circle through $P'$, $B'$, and $O'$, and $\mathcal{C}'_1$ is the line $\ell(P',R')$ because $\mathcal{C}_1$ goes through the center of the inversion $A$ . Observe that $P'$ and $B'$ are in the same half-plane bounded by $\ell(A,R')=\ell(A,R)$.

Since $\mathcal{C}_2$ does not enclose $A$, the center of the inversion, $\mathcal{C}_2$ does not enclose $R$ if and only if $\mathcal{C}'_2$ does not enclose $R'$. Then, $\mathcal{C}'_2$ does not enclose $R'$, which implies that $O'$ lies on the line segment $R'P'$. This ensures that $O'$ and $B'$, also $O$ and $B$, are in the same side of $\ell(A, R)$.

Consider now the case where $P$ is to the left of $A$ (see Figure~\ref{fig:aux1-3}). Make again a circle inversion at $A$, in which $\mathcal{C}'_1$ is the line $\ell(P',R')$ (see Figure~\ref{fig:aux1-4}). Since $\mathcal{C}_2$ encloses the center $A$ of the inversion, $\mathcal{C}_2$ does not enclose $R$ if and only if $\mathcal{C}'_2$ encloses $R'$. Then, $\mathcal{C}'_2$ encloses $R'$, which implies that $R'$ belongs to the segment $P'O'$, and also that $P'$ and $O'$ are separated by $\ell(A,R')=\ell(A,R)$. This guarantees that $O'$ and $B'$, also $O$ and $B$, are in the same half-plane bounded by $\ell(A,R)$.
\end{proof}

\begin{figure}[t]
	\centering
	\subfloat[]{
		\includegraphics[scale=0.85,page=4]{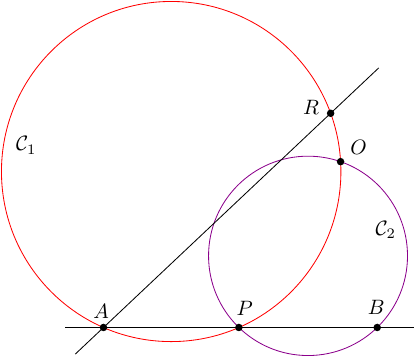}
		\label{fig:aux1-1}
	}\hspace{0.5cm}
	\subfloat[]{
		\includegraphics[scale=0.85,page=5]{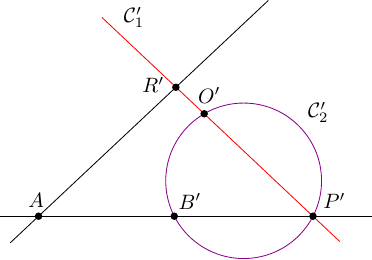}
		\label{fig:aux1-2}
	}\\
	\subfloat[]{
		\includegraphics[scale=0.85,page=6]{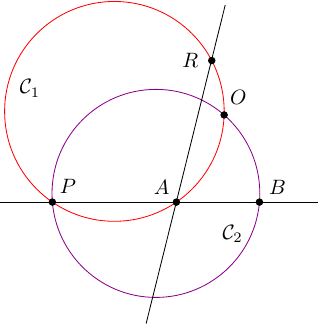}
		\label{fig:aux1-3}
	}\hspace{0.5cm}
	\subfloat[]{
		\includegraphics[scale=0.85,page=7]{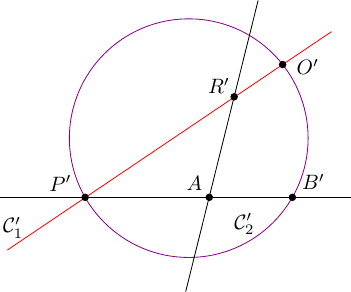}
		\label{fig:aux1-4}
	}
	\caption{\small{
		Illustration of Lemma~\ref{lem:aux1}.
	}}
	\label{fig:aux1}
\end{figure}

Finally, we need one more technical lemma, illustrated in Figure~\ref{fig:aux2}.

%\begin{lemma}\label{lem:aux2}
%Let $\ell$ be a line, and $R,C\in \ell$ two points. Let $h$ be a half-line with apex point $H$ such that the line containing $h$ intersects $\ell$ at $R$ and is perpendicular to $\ell$. Let $\delta$ be the half-plane bounded by $\ell$ such that $\delta\cap h$ is a half-line. Then, for any two points $X,Y\in h$ with $|XH|\le |YH|$, we have $D_{XC}\cap\delta \subseteq D_{YC}\cap \delta$.
%\end{lemma}

\begin{lemma}\label{lem:aux2}
Let $\ell$ be a line, and $R,C\in \ell$ two points. Let $h$ be a half-line with apex point $H$ such that the supporting line of $h$ is perpendicular to $\ell$ 
at point $R$. Let $\delta$ be the half-plane bounded by $\ell$ such that $\delta\cap h$ is a half-line. Then, for any two points $X,Y\in h$ with $|XH|\le |YH|$, we have $D_{XC}\cap\delta \subseteq D_{YC}\cap \delta$.
\end{lemma}

\begin{proof}
Consider the more general case in which $h$ and $\ell$ intersect at $R$. The other case where $h$ and $\ell$ do not intersect can be proved similarly.
We analyze three cases, depending on where the points $X$ and $Y$ lie on $h$.

i) Let $X,Y\in HR$ be two points satisfying $|XH|\le |YH|$ (see Figure~\ref{fig:aux2-1}). Then, we have $\angle RXC \le \angle RYC$. For any two points $X'\in\mathcal{C}_{XC}$ and $Y' \in\mathcal{C}_{YC}$ in the interior of $\delta$, we have $\angle RX'C=\pi-\angle RXC$ and $\angle RY'C=\pi-\angle RYC$. This implies $\angle RY'C\le \angle RX'C$, and hence $D_{XC}\cap\delta \subseteq D_{YC}\cap \delta$.

ii) Let $X,Y\in (h\setminus HR)\cup \{R\}$ be two points satisfying $|XH|\le |YH|$ (see Figure~\ref{fig:aux2-2}). For any two points $X'\in\mathcal{C}_{XC}$ and $Y'\in\mathcal{C}_{YC}$ in the interior of $\delta$, we have $\angle RY'C =\angle RYC \le \angle RXC=\angle RX'C$. This implies $D_{XC}\cap\delta \subseteq D_{YC}\cap \delta$.

iii) Finally, if $X\in HR$ and $Y\in (h\setminus HR)\cup \{R\}$, from the first case we have $D_{XC}\cap\delta \subseteq D_{RC}\cap \delta$, and from the second one $D_{RC}\cap\delta \subseteq D_{YC}\cap \delta$. Hence, $D_{XC}\cap\delta \subseteq D_{YC}\cap \delta$, and the lemma is proved.
\end{proof}

\begin{figure}[t]
	\centering
	\subfloat[]{
		\includegraphics[scale=0.95,page=8]{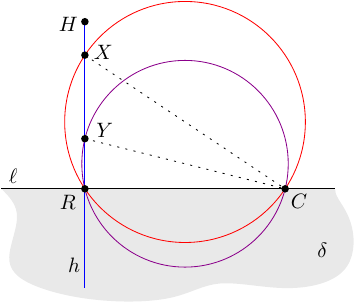}
		\label{fig:aux2-1}
	}\hspace{0.5cm}
	\subfloat[]{
		\includegraphics[scale=0.95,page=9]{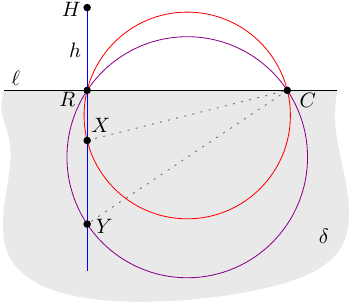}
		\label{fig:aux2-2}
	}
	\caption{\small{
		Illustration of Lemma~\ref{lem:aux2}.
	}}
	\label{fig:aux2}
\end{figure}

We have now all the tools to prove the main lemma in this work.

\begin{lemma}\label{lem:3matching}
Let $p_1, p_2,p_3\in R$ and $q_1,q_2,q_3\in B$ such that $\{(p_1,q_1),(p_2,q_2),(p_3,q_3)\}$ is a maximum matching for $\{p_1,q_1,p_2,q_2,p_3,q_3\}$. Then, the disks $D_{p_1q_1}$, $D_{p_2q_2}$, and $D_{p_3q_3}$ have a point in common.
\end{lemma}

\begin{proof}
The idea is to reduce the disks $D_{p_1q_1}$, $D_{p_2q_2}$, and $D_{p_3q_3}$ as much as possible so that each of the new three disks is contained in its corresponding original disk, and the new disks still have a point in common that is easier to find than for the original disks.
We begin by observing that a maximum matching for three pairs of points must also be maximum for any subset of two pairs, thus the implications of Lemma~\ref{lem:2-matching} must hold for any two pairs that we take.

We will shrink the three diametral disks as much as possible, while maintaining the conditions of Lemma~\ref{lem:2-matching}.
Formally, for every $\varepsilon_1\in[0,|p_1q_1|]$, $\varepsilon_2\in [0,|p_2q_2|]$, and $\varepsilon_3\in[0,|p_3q_3|]$, let $q_1(\varepsilon_1)\in p_1q_1$, $q_2(\varepsilon_2)\in p_2q_2$, and $q_3(\varepsilon_3)\in p_3q_3$ be the points such that $|q_1q_1(\varepsilon_1)|=\varepsilon_1$, $|q_2q_2(\varepsilon_2)|=\varepsilon_2$, and $|q_3q_3(\varepsilon_3)|=\varepsilon_3$. Let $(\tilde{\varepsilon}_1,\tilde{\varepsilon}_2,\tilde{\varepsilon}_3)$ be a maximal point of the set $[0,|p_1q_1|]\times [0,|p_2q_2|]\times [0,|p_3q_3|]$ such that the conditions of Lemma~\ref{lem:2-matching} are satisfied pairwise, that is, the following three statements hold:

\begin{itemize}
	\item[(1)] in the direction from $p_1$ to $p_2$, $q_2(\tilde{\varepsilon}_2)$ is not to the right of $q_1(\tilde{\varepsilon_1})$;
	\vspace{-6pt}
	\item[(2)] in the direction from $p_2$ to $p_3$, $q_3(\tilde{\varepsilon}_3)$ is not
	to the right of $q_2(\tilde{\varepsilon_2})$;
	\vspace{-6pt}	
	\item[(3)] in the direction from $p_3$ to $p_1$, $q_1(\tilde{\varepsilon}_1)$ is not to the right of $q_3(\tilde{\varepsilon}_3)$.
		\vspace{-4pt}
\end{itemize}
The point $(\tilde{\varepsilon}_1,\tilde{\varepsilon}_2,\tilde{\varepsilon}_3)$ is {\em maximal} if there does not exist any other point
$(\varepsilon'_1,\varepsilon'_2,\varepsilon'_3)\in [0,|p_1q_1|]\times [0,|p_2q_2|]\times [0,|p_3q_3|]$ such that $\tilde{\varepsilon}_1\le \varepsilon'_1$, $\tilde{\varepsilon}_2\le \varepsilon'_2$, $\tilde{\varepsilon}_3\le \varepsilon'_3$, and the above three conditions are also satisfied by using $(\varepsilon'_1,\varepsilon'_2,\varepsilon'_3)$ instead of $(\tilde{\varepsilon}_1,\tilde{\varepsilon}_2,\tilde{\varepsilon}_3)$.

\begin{figure*}[bt]
	\centering
	\subfloat[]{
		\includegraphics[scale=0.8,page=3]{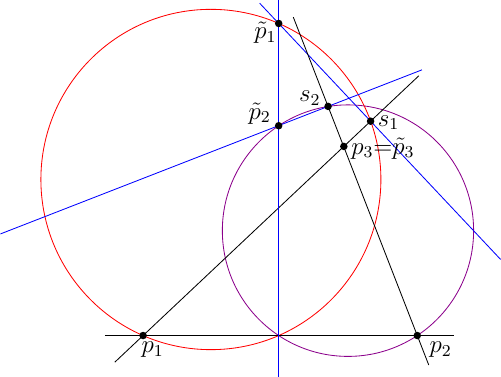}
		\label{fig:3matching1}
	}
	\subfloat[]{
		\includegraphics[scale=0.8,page=10]{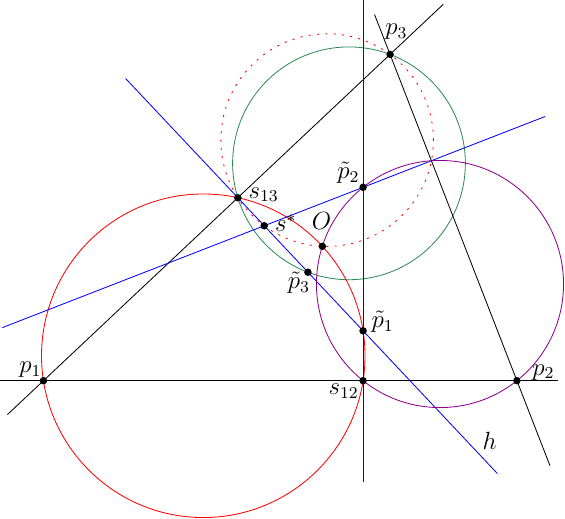}
		\label{fig:3matching2}
	}
	\caption{\small{
		Illustration of the two cases considered in the proof of Lemma~\ref{lem:3matching}:
		(a) $\tilde{p}_3=p_3$, (b) $\tilde{p}_3\neq p_3$.
	}}
	\label{fig:3matching}
\end{figure*}

Let $\tilde{p}_1=q_1(\tilde{\varepsilon}_1)$, $\tilde{p}_2=q_2(\tilde{\varepsilon}_2)$,
and $\tilde{p}_3=q_3(\tilde{\varepsilon}_3)$. Note that $D_{p_1\tilde{p}_1}\subseteq D_{p_1q_1}$, $D_{p_2\tilde{p}_2}\subseteq D_{p_2q_2}$, and $D_{p_3\tilde{p}_3}\subseteq D_{p_3q_3}$. We prove now that $D_{p_1\tilde{p}_1}$, $D_{p_2\tilde{p}_2}$, and $D_{p_3\tilde{p}_3}$ have a point in common, which implies the lemma.

If $p_1$, $p_2$, and $p_3$ belong to the same line $\ell$ (assuming w.l.o.g.\ that they appear in this order in $\ell$), then the points $\tilde{p}_1$, $\tilde{p}_2$, and $\tilde{p}_3$ belong to the same line $\ell'$ perpendicular to $\ell$ . By Thales' theorem, the point $\ell\cap\ell'$ is common to $D_{p_1\tilde{p}_1}$, $D_{p_2\tilde{p}_2}$, and $D_{p_3\tilde{p}_3}$. Hence, assume from now on that $p_1$, $p_2$, and $p_3$ are in general position. Then, there are two cases to consider:

\medskip

\noindent{\bf Case 1}: $\tilde{p}_i=p_i$ for some $i \in \{1,2,3\}$. Assume w.l.o.g. $\tilde{p}_3=p_3$ (see Figure~\ref{fig:3matching1}).
Then $D_{p_3\tilde{p}_3}$ consists of a single point, $p_3$.
Let $s_1$ and $s_2$ be the orthogonal projections of $\tilde{p}_1$ on $\ell(p_1,p_3)$, and $\tilde{p}_2$ on $\ell(p_2,p_3)$, respectively. From the third condition above, we have that in the direction from $p_3$ to $p_1$, point $\tilde{p}_1$ is not to the right of $\tilde{p}_3$.
%\rodrigo{No es al reves?}
%\pablo{Lo veo bien así: Lo que se quiere decir es que si $p_3$ está a la izquierda de $p_1$ (ambos en el eje $x$), entonces en la dirección horizontal $\tilde(p)_1$ y $\tilde(p)_3$ están en la misma vertical, o $\tilde(p)_1$ está a la izquierda de $\tilde(p)_3$.}
Thus, we have $p_3\in p_1s_1$. Similarly, since in the direction from $p_2$ to $p_3$, point $\tilde{p}_3$ is not to the right of $\tilde{p}_2$, we have $p_3\in p_2s_2$. By Thales' theorem $p_1s_1\subset D_{p_1\tilde{p}_1}$ and $p_2s_2\subset D_{p_2\tilde{p}_2}$, hence $p_3=\tilde{p}_3$ is common to $D_{p_1\tilde{p}_1}$, $D_{p_2\tilde{p}_2}$, and $D_{p_3\tilde{p}_3}$.

\medskip

\noindent{\bf Case 2}: $\tilde{p}_1\neq p_1$, $\tilde{p}_2\neq p_2$, and $\tilde{p}_3\neq p_3$. By construction of $\tilde{p}_1$, $\tilde{p}_2$, and $\tilde{p}_3$, at least two pairs of lines among $(\ell(p_1,p_2),\ell(\tilde{p}_1,\tilde{p}_2))$, $(\ell(p_2,p_3),\ell(\tilde{p}_2,\tilde{p}_3))$, and $(\ell(p_3,p_1),\ell(\tilde{p}_3,\tilde{p}_1))$ form perpendicular lines.
Note that this last statement follows from the fact that $(\tilde{\varepsilon}_1,\tilde{\varepsilon}_2,\tilde{\varepsilon}_3)$ is taken as a maximal point. That is, at least two segments among $p_1q_1(\varepsilon_1)$, $p_2q_2(\varepsilon_2)$, and $p_3q_3(\varepsilon_3)$ cannot be shortened by decreasing their corresponding values of $\varepsilon_1$, $\varepsilon_2$, and $\varepsilon_3$, so that statements (1-3) are still satisfied. For example, the extreme cases of statement~(1) are when lines $\ell(p_1,p_2)$ and $\ell(q_1(\varepsilon_1),q_2(\varepsilon_2))$ are perpendicular.

%\clemens{no me queda claro, pero no he buscado mucho}
%\pablo{Esto creo está bien, sale de mirar que al menos dos puntos entre $q1(eps_1)$, $q2(eps_2)$ y $q3(eps_3)$ no pueden 'moverse' más porque dejarían de cumplirse las condiciones del itemlist de más arriba. Por ejemplo, el límite de la primera condición del itemlist es que las rectas $\ell(p_1,p_2)$ y $\ell(q_1(eps_1), q_2(eps_2))$ sean perpendiculares. En la versión journal deberíamos explicar más }

%\rodrigo{Pablo, podrías explicar un poco mas esto entonces?}

Assume w.l.o.g.\ that $\ell(\tilde{p}_1,\tilde{p}_2)$ is perpendicular to $\ell(p_1,p_2)$, and that $\ell(\tilde{p}_3,\tilde{p}_1)$ is perpendicular to $\ell(p_3,p_1)$ (see Figure~\ref{fig:3matching2}).

Let $s_{12}=\ell(p_1,p_2)\cap \ell(\tilde{p}_1,\tilde{p}_2)$, $s_{13}=\ell(p_1,p_3)\cap \ell(\tilde{p}_1,\tilde{p}_3)$, and let $s^*$ be the point of $\ell(\tilde{p}_1,\tilde{p}_3)$ such that $\ell(s^*,\tilde{p}_2)$ is perpendicular to $\ell(p_2,p_3)$.

By Thales' theorem, $s_{12},s_{13}\in \mathcal{C}_{p_1\tilde{p}_1}$, $s_{12}\in \mathcal{C}_{p_2\tilde{p}_2}$, and $s_{13}\in \mathcal{C}_{p_3\tilde{p}_3}$. If $\mathcal{C}_{p_1\tilde{p}_1}$ and $\mathcal{C}_{p_2\tilde{p}_2}$ are tangent at $s_{12}$, let $O=s_{12}$, otherwise let $O$ be the intersection point other than $s_{12}$ between $\mathcal{C}_{p_1\tilde{p}_1}$ and $\mathcal{C}_{p_2\tilde{p}_2}$.
By Lemma~\ref{lem:3circles}, we have that $O=\mathcal{C}_{p_1\tilde{p}_1}\cap \mathcal{C}_{p_2\tilde{p}_2}\cap \mathcal{C}_{p_3s^*}$.

If $\mathcal{C}_{p_2\tilde{p}_2}$ contains $s_{13}$, then we are done, 
%since $s_{13}$ is common to $D_{p_1\tilde{p}_1}$, $D_{p_2\tilde{p}_2}$, and $D_{p_3\tilde{p}_3}$.
since $s_{13}$ is also common to $D_{p_1\tilde{p_1}}$ and $D_{p_3 \tilde{p}_3}$.
 Hence, assume $\mathcal{C}_{p_2\tilde{p}_2}$ does not contain $s_{13}$. Under this assumption, by Lemma~\ref{lem:aux1} (used with points $A=p_1$, $B=p_2$, $P=s_{12}$, and $R=s_{13}$), $O$ and $p_2$ are in the same half-plane $\mathcal{H}$ bounded by $\ell(p_1,p_3)$.

Since $\tilde{p}_3$ is not to the right of $\ell(s^*,\tilde{p}_2)$ in the direction from $p_2$ to $p_3$, we have that $\tilde{p}_3$ is on the half-line $h\subset \ell(s^*,\tilde{p}_1)$ with apex $s^*$ and such that $h\cap \mathcal{H}$ is a half-line. By Lemma~\ref{lem:aux2}, $D_{p_3s^*}\cap \mathcal{H} \subseteq D_{p_3\tilde{p}_3}\cap \mathcal{H}$, and hence $O\in D_{p_3\tilde{p}_3}$, which implies that $O$ is common to $D_{p_1\tilde{p}_1}$, $D_{p_2\tilde{p}_2}$, and $D_{p_3\tilde{p}_3}$.
\end{proof}

\begin{theorem}\label{theo:main}
Given a set $R$ of $n\ge 2$ red points and a set $B$ of $n$ blue points, in any maximum matching of $R$ and $B$, the disks have a common intersection.
\end{theorem}

%\begin{theorem}\label{theo:main}
%Given a set $R$ of $n\ge 2$ red points and a set $B$ of $n$ blue points, there exists a matching $\M$ of $R\cup B$ such that $\bigcap_{(p,q)\in \M} D_{pq} \neq \emptyset$.
%\end{theorem}

\begin{proof}
Let $\M=\{(p_1,q_1),(p_2,q_2),\ldots,(p_n,q_n)\}$ be a maximum matching of $R\cup B$. If $n=2$, then $D_{p_1q_1}\cap D_{p_2q_2}\neq \emptyset$ by Lemma~\ref{lem:pairwise-intersection}, implying the theorem. Otherwise, if $n\ge 3$, for every different $i,j,k\in\{1,2,\ldots,n\}$ the matching $\{(p_i,q_i),(p_j,q_j),(p_k,q_k)\}$ must be maximum for $\{p_i,q_i,p_j,q_j,p_k,q_k\}$. Then, by Lemma~\ref{lem:3matching}, $D_{p_iq_i}\cap D_{p_jq_j}\cap D_{p_k,q_k}\neq \emptyset$. The result follows by Helly's theorem.
\end{proof}

\section{Matching points with other shapes}

We finish by observing that the fact that disks, and not other arbitrary shapes, are used to match the pairs of points is important.
As mentioned before, it is clear that simpler shapes such as line segments do not have the property of giving always a common intersection.
Furthermore, we observe that replacing circles by somewhat similar shapes,  such as hexagons or decagons, in general does not preserve the property.
We can adapt the definition of diametral disk to regular hexagons or decagons as follows.
Define the \emph{diametral hexagon} (resp. \emph{decagon}) of a pair of points as the smallest-area regular hexagon (resp. decagon) that contains both points on its boundary.

In Fig.~\ref{fig:kgons} we show a simple construction that consists of four points on the vertices of a square, alternating colors.
The point set in the construction has two different perfect matchings, which are symmetric.
However, when the matching shape is a $k$-gon for $k=6$ or $k=10$, both matchings result in disjoint $k$-gons, so there is no common intersection.
The same situation occurs  with any regular $k$-gon where $k = 4q+2$, for any integer $q \geq 1$.
Note also that, even though the construction is presented degenerate for simplicity (i.e., the four points are cocircular), it can be perturbed while keeping the two pairs of $k$-gons disjoint.

\begin{figure*}[t]
	\centering
		\includegraphics{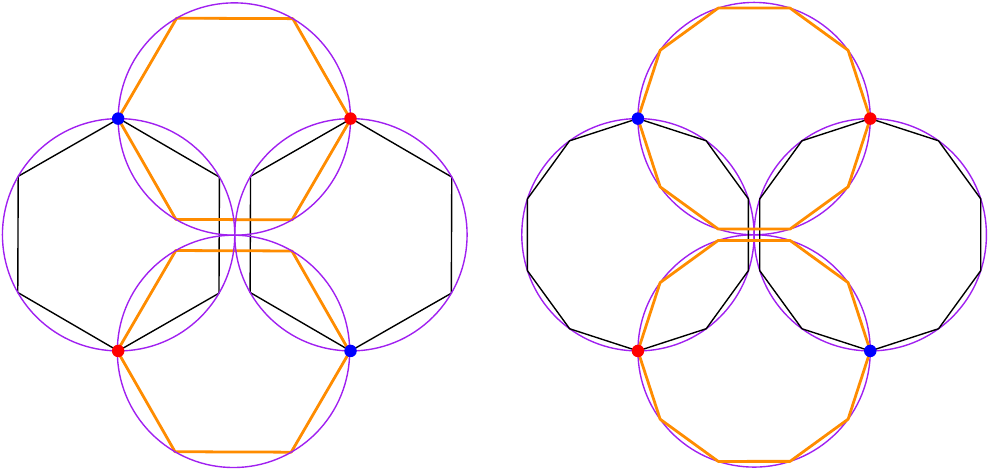}
	\caption{Construction with two pairs of points showing that using regular $k$-gons ($k=6$ on the left, $k=10$ on the right) instead of disks in many cases does not result in the property shown for disks: there is no common intersection.}
	\label{fig:kgons}
\end{figure*}

%\section{Future work}
%\rodrigo{Higher dimensions}

\small

\medskip
\noindent
\textbf{Acknowledgements.}
We thank the {\tt GeoGebra} open source software and its developers~\cite{gg}.
%\vspace{-0.6cm}
C.H., C.S., and R.S. were supported by projects Gen.\ Cat.\ 2017SGR1336, 2017SGR1640, and MINECO MTM2015-63791-R.
R.S. was also supported  by MINECO through the Ra{m\'o}n y Cajal program.
P.P.-L. was supported by projects CONICYT
FONDECYT/Regular 1160543 (Chile), and Millennium
Nucleus Information and Coordination in Networks ICM/FIC RC130003
(Chile).
\begin{figure}[h!]
\begin{minipage}[l]{0.3\textwidth}
\includegraphics[trim=10cm 6cm 10cm 5cm,clip,scale=0.15]{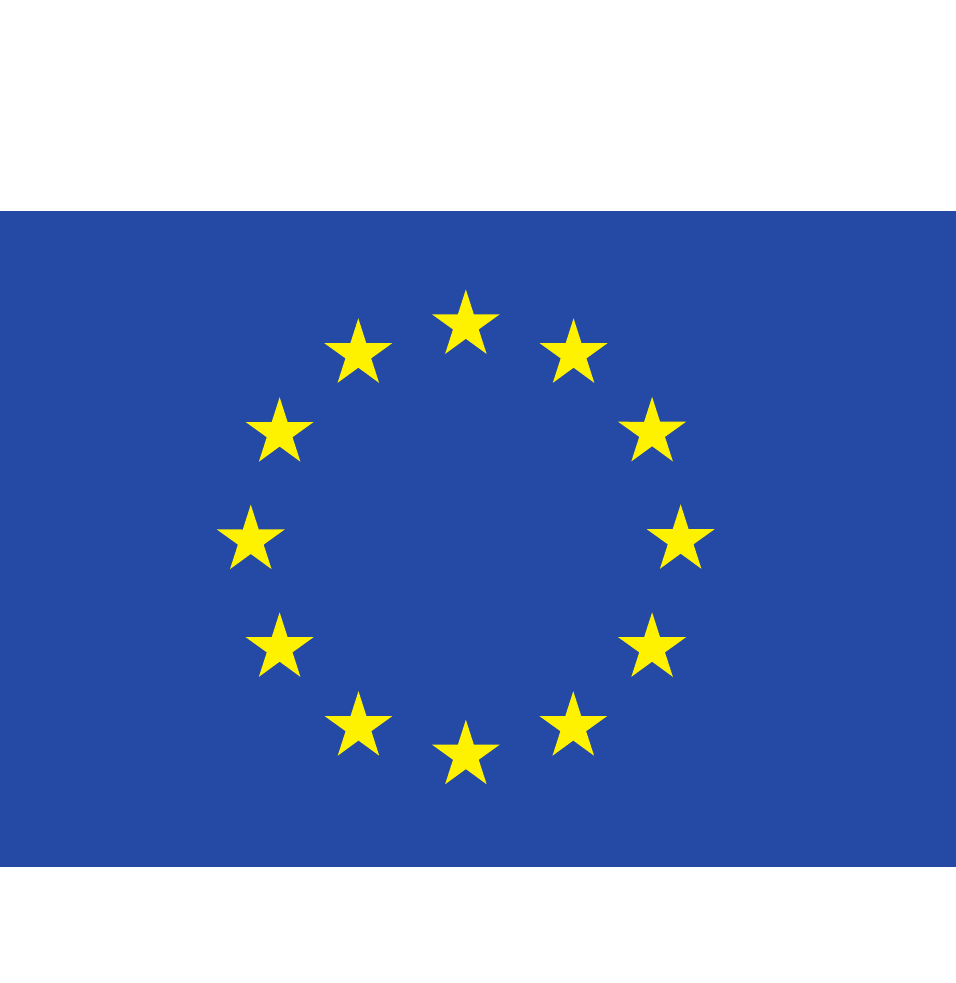}
\end{minipage}\hspace{-3.5cm}
\begin{minipage}[l][1cm]{0.9\textwidth}
This project has received funding from the European Union's Horizon 2020 research and innovation programme under the Marie Sk\l{}odowska-Curie grant agreement No 734922.
\end{minipage}
\end{figure}
%All authors have been further supported by H2020-MSCA-RISE project 734922 -- CONNECT.

%\begin{minipage}[l]{0.3\textwidth} \includegraphics[trim=10cm 6cm 10cm 5cm,clip,scale=0.15]{eu_logo} \end{minipage}  \hspace{-2cm} \begin{minipage}[l][1cm]{0.7\textwidth}
%      This project has received funding from the European Union's Horizon 2020 research and innovation programme under the Marie Sk\l{}odowska-Curie grant agreement No 734922.
% \end{minipage}

\bibliographystyle{elsarticle-num-names}

\end{document}